\documentclass[11]{article}
\usepackage{setspace}

\usepackage{arxiv}

\usepackage{amsthm,amssymb}
\usepackage{amsmath}
\usepackage{cite}
\usepackage[utf8]{inputenc} 
\usepackage[T1]{fontenc}    
\usepackage{hyperref}       
\usepackage{url}            
\usepackage{booktabs}       
\usepackage{amsfonts}       
\usepackage{nicefrac}       
\usepackage{microtype}      
\usepackage{graphicx}
\usepackage{doi}

\def\psfancypar#1#2{\begingroup\def\par{\endgraf\endgroup\lineskiplimit=0pt}
               \setbox2=\hbox{\large\sc #2}
               \newdimen\tmpht \tmpht \ht2 \advance\tmpht by \baselineskip
               \font\hhuge=Times-Bold at \tmpht
               \setbox1=\hbox{{\hhuge #1}}
               \count7=\tmpht \count8=\ht1
               \divide\count8 by 1000 \divide\count7 by \count8
               \tmpht=.001\tmpht\multiply\tmpht by \count7
               \font\hhuge=Times-Bold at \tmpht
               \setbox1=\hbox{{\hhuge #1}}
               \noindent
                \hangindent1.05\wd1
               \hangafter=-2 {\hskip-\hangindent
               \lower1\ht1\hbox{\raise1.0\ht2\copy1}%
                \kern-0\wd1}\copy2\lineskiplimit=-1000pt}

\newcommand{\beq}{\begin{equation}}
\newcommand{\eeq}{\end{equation}}
\newcommand{\bqa}{\begin{eqnarray}}
\newcommand{\eqa}{\end{eqnarray}}
\newcommand{\bqn}{\begin{eqnarray*}}
\newcommand{\eqn}{\end{eqnarray*}}

\newcommand{\be}{\begin{enumerate}}
\newcommand{\ee}{\end{enumerate}}
\newcommand{\bi}{\begin{itemize}}
\newcommand{\ei}{\end{itemize}}
\newcommand{\bd}{\begin{description}}
\newcommand{\ed}{\end{description}}
\newcommand{\ba}{\begin{array}}
\newcommand{\ea}{\end{array}}
\newcommand{\bde}{\begin{definition}}
\newcommand{\ede}{\end{definition}}
\newcommand{\bex}{\begin{example}}
\newcommand{\eex}{\end{example}}

\def\boxit#1{\vbox{\hrule\hbox{\vrule\kern3pt
        \vbox{\kern3pt#1\kern3pt}\kern3pt\vrule}\hrule}}

\def\reals{ { {\rm  I \kern-0.15em R }  } }
\def\complex{ {\,{{\rm C} \kern-0.50em \raise0.20ex {  |}}\, }}

\def\0bf{{\bf 0}}
\def\1bf{{\bf 1}}
\def\2bf{{\bf 2}}
\def\3bf{{\bf 3}}
\def\4bf{{\bf 4}}
\def\5bf{{\bf 5}}
\def\6bf{{\bf 6}}
\def\7bf{{\bf 7}}
\def\8bf{{\bf 8}}
\def\9bf{{\bf 9}}

\def\Rbf{{\bf R}}

\def\Nmat{\mathcal{N}}

\newtheorem{theorem}{Theorem}

\newtheorem{lemma}{Lemma}
\newtheorem{definition}{Definition}
\newtheorem{corollary}{Corollary}
\newtheorem{example}{Example}

\def\Rxx{\Rbf_{\ssstyle X\kern-.1em X}}

\let\ssstyle=\scriptscriptstyle

\def\Kout{\setbox1=\hbox{\Huge\bf K}\hbox to
1.05\wd1{\hspace{.05\wd1}
}}
\def\Sout{\setbox1=\hbox{\Huge\bf S}\hbox to 1.05\wd1{\hspace{.05\wd1}
}}

\newtheorem{algorithm}{Algorithm}

\title{Divergence Maximizing Linear Projection for Supervised Dimension Reduction}

\author{%
  Biao Chen \\
  Department of EECS\\
  Syracuse University\\
  Syracuse, NY 13244\\
  \texttt{bichen@syr.edu} \\
  \And
  Joshua Kortje\\
  Department of EECS \\
  Syracuse University\\
  Syracuse, NY 13244 \\
  \texttt{jmkortje@syr.edu} \\
}

\renewcommand{\shorttitle}{} 

\begin{document}
\maketitle

\begin{abstract}  
This paper proposes two linear projection methods for supervised dimension reduction using only the first and second-order statistics. The methods, each catering to a different parameter regime, are derived under the general Gaussian model by maximizing the Kullback-Leibler divergence between the two classes in the projected sample for a binary classification problem. They subsume existing linear projection approaches developed under simplifying assumptions of Gaussian distributions, such as these distributions might share an equal mean or covariance matrix. As a by-product, we establish that the multi-class linear discriminant analysis, a celebrated method for classification and supervised dimension reduction, is provably optimal for maximizing pairwise Kullback-Leibler divergence when the Gaussian populations share an identical covariance matrix. For the case when the Gaussian distributions share an equal mean, we establish conditions under which the optimal subspace remains invariant regardless of how the Kullback-Leibler divergence is defined, despite the asymmetry of the divergence measure itself. Such conditions encompass the classical case of signal plus noise, where both the signal and noise have zero mean and arbitrary covariance matrices. Experiments are conducted to validate the proposed solutions, demonstrate their superior performance over existing alternatives, and illustrate the procedure for selecting the appropriate linear projection solution.

\end{abstract}


\section{Introduction}

Dimension reduction has long been recognized as a key component in various data-driven learning applications, from clustering \cite{Cohen:15} to classification \cite{Wang:14} to pattern recognition \cite{Huang:19}. The most celebrated technique for dimension reduction is the Principal Component Analysis (PCA) which searches for a linear subspace that contains the highest data variability. Its success has been documented extensively in the literature for visualization and improved inference performance (see, e.g., \cite{Jolliffe:book}), especially when data or feature vectors have large dimensions relative to the sample size. 

PCA, in its original form, is an unsupervised approach for dimension reduction, i.e., labels of data are not considered even if they are available. It is interesting to note that while the foundation of PCA was laid out in \cite{Hotelling:1933}, the basic principle behind PCA was first described in \cite{Pearson:1901} under a supervised setting: dependent variables, i.e., `labels', are used to motivate the search for linear combinations of independent variables. While the development of PCA and other dimension reduction methods have largely followed the unsupervised setting, there has been increasing interest in developing supervised dimension reduction (SDR) techniques where the labels are accounted for when finding low-dimensional representations. Examples include supervised PCA either through the use of kernel trick \cite{Barshan:11} or manifold learning \cite{Ritchie:19}, non-negative matrix factorization \cite{Lee:99}, and supervised autoencoder \cite{Le:18}. See \cite{Chao:19} for a recent overview in this area.

As with many other inference problems, there is a strong incentive to develop SDR under the Gaussian setting. While the Gaussian model may not be universally applicable to various learning problems, understanding the solution to the Gaussian model often brings insights into the solution of the general problem. For example, linear discriminant analysis (LDA), originally derived as the optimal Bayesian boundary for classification problems with Gaussian samples \cite{Hastie:book}, has been successfully applied to a wide range of data analysis problems that are not necessarily Gaussian. This is due to its simplicity as it requires knowledge of only the first and second-order statistics and has a good generalization property, and versatility for data visualization and preprocessing. 

It is perhaps surprising that under the most general Gaussian model, the optimal low-dimensional representation is still largely unsolved for a binary classification problem, arguably the simplest supervised learning problem. There have been several notable attempts in recent years; however, they were invariably developed under some simplifying assumptions, e.g., the Gaussian models share identical mean or covariance matrix for the two classes. 
We note that LDA itself was developed under the common covariance matrix assumption and has long been recognized as a natural, {\em albeit} restrictive, way for dimension reduction for classification problems\footnote{While quadratic discriminant analysis (QDA) \cite{Hastie:book} was developed under such a general assumption, it is not a linear projection and loses its appeal to be used as a dimension reduction technique.}. The restriction is reflected in the rank of the projected subspace - with $K$ classes, LDA, or more precisely $K$-class LDA necessarily leads to a $(K-1)$-dimensional subspace. With $K=2$, i.e., a binary classification problem, LDA leads to a rank-$1$ subspace when used as a dimension reduction method. To relax this restriction, a linear projection, referred to as the Linear Optimal Low-rank (LoL) projection was proposed in \cite{Vogelstein:21}. LoL essentially concatenates the mean difference vector with the principal components of the common covariance matrix. The reduced dimension representation is shown to achieve a Chernoff Information (CI) larger than standard PCA. The other case when the Gaussian distributions share an identical (or equivalently, zero) mean with different covariance matrices was considered in  \cite{Dwivedi:21,Dwivedi:22}. Resorting to the Poincar\'{e} Separation Theorem, the authors established that a linear subspace can be constructed to maximize the Kullback-Leibler divergence (KLD) between the two distributions for the projected low-dimensional sample. Furthermore, by putting the first distribution in an isotropic position, the subspace is spanned by the eigenvectors of the second covariance matrix corresponding to a set of eigenvalues.

This paper explores SDR for binary classifications under the general Gaussian model, i.e., the distributions for different classes do not share an equal mean or covariance matrix. We propose two linear projection solutions, each catering to different parameter regimes of the Gaussian distributions, for maximizing the KLD with the low-dimensional representation. KLD has been widely used as a proxy for inference performance in various machine learning problems due to its tractability as an optimization metric \cite{Claici:20,Ji:22,Zhang:23}. It is also known to be the optimal error exponent of the type II error probability for a binary hypothesis testing problem, a.k.a., Stein's lemma \cite{Cover:book}. 

For the so-called large-$\mu$ regime, defined as the case when the mean separation between the classes is large relative to the difference in covariance matrices, our solution subsumes the LDA as its special case when the Gaussian distributions share an identical covariance matrix. As a by-product, we show that the $K$-class LDA provably maximizes pairwise KLD in the reduced $(K-1)$-dimensional subspace. 
For the small-$\mu$ regime, i.e., when the mean separation is not more than the difference in covariance matrices, the proposed solution strives for a balance in preserving the discriminative information in the mean vectors and the covariance matrices. It reduces to the optimal solution proposed in \cite{Dwivedi:22} when the Gaussian distributions share an equal mean. The proposed solution allows us to partially answer a question often left unaddressed when KLD is used as a proxy for inference performance: whether the projected subspace is invariant to the order with which the KLD is defined, as KLD itself is not symmetric in the two distributions with unequal covariance matrices. We establish conditions under which the reduced subspace remains invariant regardless of how KLD is defined. 
The conditions encompass the classical case of the `signal plus noise' model when both signal and noise are Gaussian variables with zero mean and arbitrary covariance matrices. 

Experiments using synthetic data validate our proposed approaches, both in terms of visualization and inference performance. Gradient descent algorithms using the Adam Optimizer are also implemented. Empirical results indicate that the proposed algorithms are close to global optima especially when the dimension of the subspace is not too small. The experiments also illustrate the procedure for determining the parameter regime to guide the selection of the appropriate linear projection solution. 

The rest of the paper is organized as follows. Section~\ref{sec:formulation} introduces the problem formulation. Existing linear projection approaches under the Gaussian models are reviewed in Section~\ref{sec:review} followed by some new observations in Section~\ref{sec:observations}. Section~\ref{sec:new} describes the two proposed linear projection solutions for the general Gaussian model. The multi-class SDR is treated in Section~\ref{sec:K} where we establish the optimality of multi-class LDA for maximizing pairwise KLD under the common covariance matrix model. Numerical experiments are provided in Section~\ref{sec:numerical} to demonstrate the effectiveness of the proposed SDR linear projection solutions, both for data visualization and classification performance. 
Section~\ref{sec:conclusion} concludes the paper and points out directions for future work.

\section{Problem Formulation \label{sec:formulation}}

Let $X\in R^d$ be a $d$-dimensional feature vector for a classification problem with $K$ classes. Under class $k$, $X$ is distributed according to $p_k\sim \mathcal{N}(\mu_k,\Sigma_k)$, for $k=1,\cdots,K$. We do not make any restrictive assumption on the parameters $(\mu_k,\Sigma_k)$ for any $k$. The objective is to find, for a given $r<d$, a linear projection from $R^d$ to $R^r$ to optimize, in some manner, the classification performance using only the reduced dimension sample in $R^r$. 

We will primarily consider the case with $K=2$, i.e., when the choice is binary. Extension to multi-class classification will be considered under the common covariance matrix assumption. With $K=2$, a single metric can be used to drive the search for low-dimensional representation. We use the KLD between $p_1$ and $p_2$ as the optimization criterion, defined as \cite{Kullback:51}
\[
D(p_1\|p_2)=E_{p_1}\left[ \log \frac{p_1(X)}{p_2(X)}\right],
\]
where $E_{p_1}[\cdot]$ denotes expectation with respect to $p_1$. The KLD between two multivariate Gaussian distributions $p_1\sim\mathcal{N}(\mu_1,\Sigma_1)$ and $p_2\sim\mathcal{N}(\mu_2,\Sigma_2)$ is
\begin{equation}
    D(p_1\|p_2)=\frac12 \left[\ln \frac{|\Sigma_2|}{|\Sigma_1|} - d + \mbox{tr}(\Sigma_2^{-1}\Sigma_1)+(\mu_2-\mu_1)^T\Sigma_2^{-1}(\mu_2-\mu_1)\right], \label{eq:kld}
\end{equation}
where $|\cdot|$ and tr($\cdot$) denote the 
determinant and trace of a square matrix. We assume throughout the paper that both $\Sigma_1$ and $\Sigma_2$ are full rank, hence the KLD in (\ref{eq:kld}) is well defined.  

Dimension reduction, in the context of the above inference problem, refers to finding a low-dimensional representation of the samples with which the inference can be carried out optimally. The most widely adopted approach under the Gaussian model is the linear projection-based dimension reduction. Specifically, for some positive integer $r<d$, the task is to design an $r\times d$ matrix $A$ of rank $r$ such that the inference, e.g., classification, can be carried out using $y=Ax$ instead of the original sample $x$.

 With $y=Ax$, the reduced dimension representation $y$ follows either $q_1\sim \mathcal{N}(A\mu_1,A\Sigma_1A^T)$ or $q_2\sim \mathcal{N}(A\mu_2,A\Sigma_2A^T)$ and the corresponding KLD is 
\begin{equation}
    D(q_1\|q_2)=\frac12 \left[\ln \frac{|A\Sigma_2A^T|}{|A\Sigma_1A^T|} - r + \mbox{tr}((A\Sigma_2A^T)^{-1}A\Sigma_1A^T) +(\mu_2-\mu_1)^TA^T(A\Sigma_2A^T)^{-1}A(\mu_2-\mu_1)\right]. \label{eq:reducedkld}
\end{equation}
The optimal $A$ is the one that maximizes (\ref{eq:reducedkld}). However, directly searching for $A$ that maximizes (\ref{eq:reducedkld}) can be daunting when $d$ is large. 
Our goal is to derive easy-to-implement linear projection solutions that are provably optimal in maximizing KLD under certain conditions. Perhaps more importantly, while it is tempting to devise a numerical procedure to find $A$ that maximizes (\ref{eq:reducedkld}), such an approach loses appeal when dealing with datasets that are not necessarily Gaussian and, as (\ref{eq:reducedkld}) is not a concave function of $A$, is susceptible to local optima even with Gaussian data. On the other hand, a simple and interpretable linear projection solution derived under the Gaussian assumption often has much better generalization property \cite{Abu-Mostafa:book} and is also easier to adapt for much broader applications when the data deviates from Gaussian. Finally, extensive numerical experiments suggest that the solutions of the proposed algorithms are likely to be close to the global optimum - gradient descent algorithms with the proposed solutions as initialization always outperform that using random initialization in all the cases we have tested.

The KLD, as with any meaningful information measures, obeys the data processing inequality (DPI) \cite{Cover:book}. Therefore, for any $A\in R^{r\times d}$,
\[
D(q_1\|q_2)\leq D(p_1\|p_2).
\]

A direct consequence of the DPI for the KLD is the following lemma which states that in looking for the optimal $A$ matrix, we can limit our search to the set of orthonormal matrices\cite{Dwivedi:22}. 
\begin{lemma} [Theorem 1 \cite{Dwivedi:22}] \label{lem:matrix}
    For any $r\times d$ matrix $A$ of rank $r$, there exists an $r\times d$ orthonormal matrix $B$, i.e., the rows of $B$ are orthonormal vectors, such that the resulting KLD is identical.  
\end{lemma}

We note that Lemma~\ref{lem:matrix} is a direct consequence of the DPI for KLD since for any rank-$r$ $r\times d$ matrix, there always exist a full rank $r\times r$ matrix $T$ and an orthonormal $r\times d$ matrix $B$ such that $A=TB$ (e.g., through QR factorization of $A$). Let $D(q_1\|q_2)$ be the KLD of the reduced dimension sample using the projection matrix $A$ whereas $D(h_1\|h_2)$ be the KLD of the reduced dimension sample using $B$. Then with $A=TB$, $D(q_1\|q_2)\leq D(h_1\|h_2)$. But $B=T^{-1}A$, hence $D(h_1\|h_2)\leq D(q_1\|q_2)$. 

We do not consider the choice of $r$ in this paper. Rather, the projection solution is provided for any given $r<d$. In real applications, cross-validation can be employed to select $r$ as typically done for any hyperparameters in learning problems.

\section{Review of Existing Approaches \label{sec:review}}
Existing linear projection approaches for SDR with Gaussian data were invariably developed under a simplified model: the two Gaussian distributions share either 1) an equal covariance matrix or 2) an identical mean.

\subsection{Equal Covariance Matrix}
With $\Sigma_1=\Sigma_2\triangleq \Sigma$, the KLD between the two Gaussian distributions reduces to 
\begin{equation}
    D(p_1\|p_2)=\frac12 (\mu_2-\mu_1)^T\Sigma^{-1}(\mu_2-\mu_1). \label{eq:kldequalSigma}
\end{equation}

We note that the KLD for this case is symmetric between the two Gaussian distributions, i.e., $D(p_1\|p_2)=D(p_2\|p_1)$.
The most well-known SDR approach for this special case is LDA. For the binary case, LDA projects the sample in the direction of $a=\Sigma^{-1}(\mu_2-\mu_1)$, resulting in the statistic $y=(\mu_2-\mu_1)^T\Sigma^{-1} x$. With an equal covariance matrix, this statistic is equivalent to the likelihood ratio of the two Gaussian distributions, the very reason that this is the optimal statistic for the underlying binary hypothesis testing problem.

We digress here to mention that the generalization of the two-class LDA described above to the multi-class, i.e., the multi-class LDA, is often used as a dimension reduction method. However, the dimension of reduced subspace needs necessarily to be of no larger than $K-1$ for $K$ classes (c.f. Section~\ref{sec:K}). This often puts an undesired restriction when using it as an SDR method. Relaxing this restriction, the authors in \cite{Vogelstein:21} suggested concatenating the mean difference vectors with the principal components of the common covariance matrix $\Sigma$. Specifically, for the case of $K=2$, the $r$-dimensional subspace can be found by concatenating $\mu_1-\mu_2$ along with $r-1$ principal components corresponding to $\Sigma$. We note that \cite{Vogelstein:21} uses CI as the optimizing criterion. The CI between two multivariate normal distributions, $p_1\sim \mathcal{N}(\mu_1,\Sigma_1)$ and $p_2\sim \mathcal{N}(\mu_2,\Sigma_2)$ is 
\[
C(p_1,p_2) = \max_{0\leq s \leq 1} \left(\frac{s(1-s)}{2}(\mu_2-\mu_1)^T\Sigma_s^{-1}(\mu_2-\mu_1)+\frac12 \log \frac{|\Sigma_s|}{|\Sigma_1|^s|\Sigma_2|^{1-s}}\right),
\]
where $\Sigma_s=s\Sigma_1+(1-s)\Sigma_2$.
Under the equal covariance matrix assumption, the CI between two Gaussian distributions is reduced to
\begin{equation}
    C(p_1, p_2)=\frac18 (\mu_2-\mu_1)^T\Sigma^{-1}(\mu_2-\mu_1). \label{eq:CI}
\end{equation}
Clearly, with an equal covariance matrix,  maximizing CI is equivalent to maximizing KLD.

\subsection{Equal Mean}

With $\mu_1=\mu_2\triangleq \mu$, 
the corresponding KLD becomes 
\begin{equation}
    D(p_1\|p_2)=\frac12 \left[\ln \frac{|\Sigma_2|}{|\Sigma_1|} - d + \mbox{tr}(\Sigma_2^{-1}\Sigma_1)\right]. \label{eq:kldequalmu}
\end{equation}

With $\Sigma_1$ being full rank, the linear transformation $x'=\Sigma_1^{-\frac12}(x-\mu)$ results in two Gaussian distributions $p_1'\sim \mathcal{N}(0,I)$ and $p_2'\sim \mathcal{N}(0,\tilde\Sigma)$ where $\tilde\Sigma=\Sigma_1^{-\frac12} \Sigma_2 \Sigma_1^{-\frac12}$. 
The resulting KLD is 
\[
D(p_1'\|p_2')=\frac12 \left[\ln |\tilde\Sigma| - d + \mbox{tr}(\tilde\Sigma^{-1})\right]. \label{eq:kld2}
\]
Using the Poincar\'{e} separation theorem\cite{Magnus:book}, it was established in \cite{Dwivedi:22}  that the optimal subspace in $R^r$, i.e., that achieves the maximum KLD between the two distributions of the projected sample, is spanned by the eigenvectors of $\tilde\Sigma$ with the corresponding eigenvalues maximizing 
\begin{equation}
    g(\lambda)=\frac12 \left(\ln \lambda-1+\frac{1}\lambda\right).
\label{eq:g}
\end{equation} 
The resulting KLD in its $R^r$ representation is 
\[
D(q_1\|q_2) = \sum_{i=1}^r g(\gamma_i).
\]
where $\gamma_i$'s are the $r$ eigenvalues maximizing $g(\gamma)$. It is interesting to note that, in contrast with PCA, the eigenvalues are not necessarily the $r$ largest ones of $\tilde\Sigma$. This is intuitive given that the task here is to discriminate between the two classes, not preserve the variability of any one particular class. Given that $p_1'$ is now isotropic in all directions, the direction favorable for discriminating the two classes may not be the direction in which $p_2'$ has the largest variance. In the special case when the eigenvalues of $\tilde\Sigma$ are all less than $1$, the smallest eigenvalues are indeed the ones that maximize (\ref{eq:g}) \cite{Dwivedi:22}. 

Since $\Sigma_1\neq \Sigma_2$, $D(p_1\|p_2)\neq D(p_2\|p_2)$ as (\ref{eq:kldequalmu}) is not symmetric in the pair $(\Sigma_1,\Sigma_2)$. An interesting question left unanswered is that if $D(p_2\|p_1)$ is used instead of $D(p_1\|p_2)$ as the optimizing criterion, will the resulting subspace be different?

\section{New Observations on the Special Cases \label{sec:observations}}

Before presenting solutions to the general case when each class has distinct mean vectors and covariance matrices, we first provide some interesting observations to the two special cases reviewed in Section~\ref{sec:review}. These observations help motivate our proposed solutions for the general case. 

\subsection{Equal Covariance Matrix}
The LDA is derived by finding the optimal decision boundary under the Bayesian inference framework for Gaussian samples with an equal covariance matrix. We establish that for the case with two classes, the LDA is also the optimal dimension reduction method if the objective is to maximize the KLD for the low-dimensional data. 
In fact, the low-dimensional representation found through LDA preserves the entire KLD, and consequently, the CI since the two metrics are equivalent (save for a scaling factor of $4$) for the equal covariance case. Therefore, there is no new information to be gained by including additional dimensions to the subspace. We emphasize again that this is only true with the equal covariance matrix assumption. An intuitive explanation for this phenomenon is that with an equal covariance matrix, both Gaussian distributions can be put into an isotropic position through the same unitary transformation (the eigenvectors of the common covariance matrix). Thus, there is no discrimination information along any direction except along the mean difference vector of the two distributions when put in an isotropic position. Indeed, for this special case, also considered in \cite{Vogelstein:21}, a one-dimensional projection suffices for the binary classification problem and outperforms the LoL in achievable KLD or CI. 

\begin{lemma}
    With $\Sigma_1=\Sigma_2=\Sigma$, a simple one-dimensional projection preserves the entire CI or KLD. 
    Furthermore, this one-dimensional projection coincides with the LDA direction $\Sigma^{-1}(\mu_2-\mu_1)$.   \label{lem:LDA}
\end{lemma}

\begin{proof}

To show that the one-dimension projection is optimal for maximizing KLD or CI, choose $a=\Sigma^{-1}(\mu_2-\mu_1)$ as the projection direction. Thus $y=a^Tx=(\mu_2-\mu_1)\Sigma^{-1}x$ is precisely the same statistic obtained via LDA, and the resulting two Gaussian distributions are respectively $q_1\sim \mathcal{N}(a^T\mu_1,a^T\Sigma a)$ and $q_2\sim\mathcal{N}(a^T\mu_2,a^T\Sigma a)$. Substitute $a$ with $\Sigma^{-1}(\mu_2-\mu_1)$, we have 
\begin{align*}
    q_1 &\sim \mathcal{N}((\mu_2-\mu_1)^T\Sigma^{-1}\mu_1,(\mu_2-\mu_1)^T\Sigma^{-1}(\mu_2-\mu_1)), \\
    q_2 &\sim \mathcal{N}((\mu_2-\mu_1)^T\Sigma^{-1}\mu_2,(\mu_2-\mu_1)^T\Sigma^{-1}(\mu_2-\mu_1)).
\end{align*}
Substitute the means and variances in (\ref{eq:kld}) with that corresponding to $q_1$ and $q_2$, we obtain
\begin{align*}
    D(q_1\|q_2)&= \frac12\frac{ ((\mu_2-\mu_1)^T\Sigma^{-1}\mu_2-(\mu_2-\mu_1)^T\Sigma^{-1}\mu_1)^T((\mu_2-\mu_1)^T\Sigma^{-1}\mu_2-(\mu_2-\mu_1)^T\Sigma^{-1}\mu_1)}{(\mu_2-\mu_1)^T\Sigma^{-1}(\mu_2-\mu_1)} \\
    &=\frac12 (\mu_2-\mu_1)^T\Sigma^{-1}(\mu_2-\mu_1).
\end{align*}
This completely recovers (\ref{eq:kldequalSigma}). Therefore, it preserves the KLD of the original sample. Given that the KLD satisfies DPI, any other projection will not yield a larger KLD. The same statement applies to CI. 
\end{proof}

The result is not surprising - LDA is equivalent to the likelihood ratio statistic and was derived as the optimal Bayesian test for the equal covariance matrix case thus it ought to be optimal in the error exponent both in the Bayesian setting (CI) and the Neyman-Pearson setting (KLD) for the binary hypothesis testing. What is significant is that this result can indeed be generalized to the $K$-class case where one can construct a $(K-1)$-dimensional linear subspace that preserves all pairwise KLD when $K>2$ classes share an equal covariance matrix (c.f. Section~\ref{sec:K}).
 
\subsection{Equal Mean}

Consider the other special case when the Gaussian distributions share the same mean vector but different covariance matrices. The procedure described in \cite{Dwivedi:22} can be shown to be equivalent to solving a generalized eigen decomposition problem. 
\begin{lemma}
    For any $1\leq r < d$, the $r$-dimensional subspace that maximizes the KLD $D(q_1\|q_2)$ is spanned by $r$ generalized eigenvectors of the matrix pair $(\Sigma_2,\Sigma_1)$, corresponding to the $r$ generalized eigenvalues that maximize $g(\cdot)$ in (\ref{eq:g}). \label{lem:generalized}
\end{lemma}
\begin{proof}
    Generalized eigen decomposition of a pair of square matrices $A$ and $B$ is expressed as 
\[
A v_i = \lambda_i B v_i
\]
where $(\lambda_i,v_i)$ are the generalized eigenvalue and eigenvector pair. 

For the pair of Gaussian distributions $p_1\sim \mathcal{N}(0,\Sigma_1)$ and $p_1\sim \mathcal{N}(0,\Sigma_2)$, the optimal $r$-dimensional subspace was found in \cite{Dwivedi:22}, described below.
\begin{enumerate}
    \item Diagonalize $\Sigma_1$ by pre-multiplying the sample with $\Sigma_1^{-\frac12}$.
    \item Compute the eigen decomposition of $\tilde\Sigma=\Sigma_1^{-\frac12}\Sigma_2\Sigma_1^{-\frac12}=U\Lambda U^T$.
    \item Find the $r$ eigenvalues in $\Lambda$ that yields the largest value of 
    \[
    g(\lambda)=\frac12 (\ln \lambda-1+\frac{1}\lambda).
    \]
    \item Choose the corresponding eigenvectors in $U$ to form the projection matrix to the optimal subspace.
\end{enumerate}

Any eigenvalue/eigenvector pair $(\lambda,u)$ for the matrix $\tilde\Sigma$ satisfies
\[
\Sigma_1^{-\frac12}\Sigma_2\Sigma_1^{-\frac12} u = \lambda  u
\]
Define $v=\Sigma_1^{-\frac12}u$. Then 
\[
\Sigma_1^{-\frac12}\Sigma_2 v = \lambda \Sigma_1^{\frac12} v
\]
Pre-multiply both sides by $\Sigma_1^{\frac12}$, we have 
\[
\Sigma_2 v = \lambda \Sigma_1 v
\]
Therefore, the eigenvalues of $\tilde\Sigma$ are precisely the generalized eigenvalues of the matrix pair $(\Sigma_2,\Sigma_1)$. Thus maximizing $g(\lambda)$ results in the same set of generalized eigenvalues of $(\Sigma_2,\Sigma_1)$. Furthermore, the eigenvectors of $\tilde\Sigma$ (the vectors $u_i$, $i=1,\cdots,d$)  and generalized eigenvectors of $(\Sigma_2,\Sigma_1)$ (the vectors $v_i$, $i=1,\cdots,d$) are 1-1 correspondence with $u_i=\Sigma_1^{\frac12}v_i$. With the transformed sample $y=\Sigma_1^{-\frac12}(x-\mu_1)$, the $r$-dimensional representation is found by projecting $y$ using the $r\times d$ matrix whose rows are a set of $u_i$'s. This is equivalent to projecting the original sample $x$ using $\Sigma_1^{-\frac12} u_i$, which is precisely $v_i$. Thus for the original sample $x$, the projected subspace is spanned by the corresponding set of generalized eigenvectors of $(\Sigma_2,\Sigma_1)$.
\end{proof}
Reformulating the solution using generalized eigen decomposition helps us to gain insight into the structure of the optimal low-dimensional subspace. In particular, for the case with unequal covariance matrices, $D(p_1\|p_2)\neq D(p_2\|p_1)$.
If one uses $D(q_2\|q_1)$ as the divergence to be maximized, the subspace will then be the $r$-dimensional subspace spanned by $r$ generalized eigenvectors of $(\Sigma_1,\Sigma_2)$, i.e., with the roles of $\Sigma_1$ and $\Sigma_2$ swapped. However, from the definition of generalized eigen decomposition, the generalized eigenvalues are reciprocal of each other and the generalized eigenvectors are identical to each other when the order of the two matrices is reversed. It remains to show under what conditions the selected eigenvalues, i.e., those maximizing $g(\lambda)$ defined in (\ref{eq:g}) correspond to the same set of eigenvectors. We have the following result.

\begin{theorem}
    If the generalized eigenvalues are all greater than $1$ or all smaller than $1$, then maximizing either $D(q_1\|q_2)$ or $D(q_2\|q_1)$ results in an identical subspace. \label{thm:order}
\end{theorem}
\begin{proof}
   Lemma~\ref{lem:generalized} states that if two Gaussian distributions share the same mean vector, then the optimal subspace of dimension $r$, i.e., one that maximizes $D(q_1\|q_2)$, can be found by solving the generalized eigen decomposition problem
\begin{equation}
    \Sigma_2 v = \lambda \Sigma_1 v. \label{eq:general1}
\end{equation}
Apply Lemma~\ref{lem:generalized} to maximize $D(q_2\|q_1)$, i.e., with the role of $q_1$ and $q_2$ reversed, the optimal subspace can now be found by solving the following generalized eigen decomposition problem 
\begin{equation}
    \Sigma_1 v' = \lambda' \Sigma_2 v'. \label{eq:general2}
\end{equation}

Comparing (\ref{eq:general1}) with (\ref{eq:general2}), it is clear that any eigenvalue/eigenvector pair $(\lambda,v)$ for $(\Sigma_2,\Sigma_1)$ has a corresponding eigenvalue/eigenvector pair $(\lambda',v')$ for $(\Sigma_1,\Sigma_2)$ with $\lambda=1/\lambda'$ and $v=v'$. That is, the eigenvectors are identical with each other while the eigenvalues are reciprocal of each other.

Consider the following two cases.
\begin{enumerate}
    \item $\lambda_i\geq 1$ for all $1\leq i \leq d$. Thus all the generalized eigenvalues for $(\Sigma_2,\Sigma_1)$ are greater than $1$. From the expression $g(\lambda)$ in (\ref{eq:g}), plotted in Fig.~\ref{fig:g}, it is trivial to verify that it is monotone increasing in $\lambda$ for $\lambda\geq 1$. Thus one should choose the largest $r$ eigenvalues and the corresponding eigenvectors form the optimal subspace. 
    \item $\lambda_i\leq 1$ for all $1\leq i \leq d$. Thus all the generalized eigenvalues for $(\Sigma_2,\Sigma_1)$ are less than $1$. From the expression $g(\lambda)$ in (\ref{eq:g}), it is again trivial to verify that it is monotone decreasing in $\lambda$ for $0< \lambda\leq 1$. Thus one should choose the smallest $r$ eigenvalues and the corresponding eigenvectors form the optimal subspace. 
\end{enumerate}
\begin{figure}[h]
    \centering
    \includegraphics[width=3in]{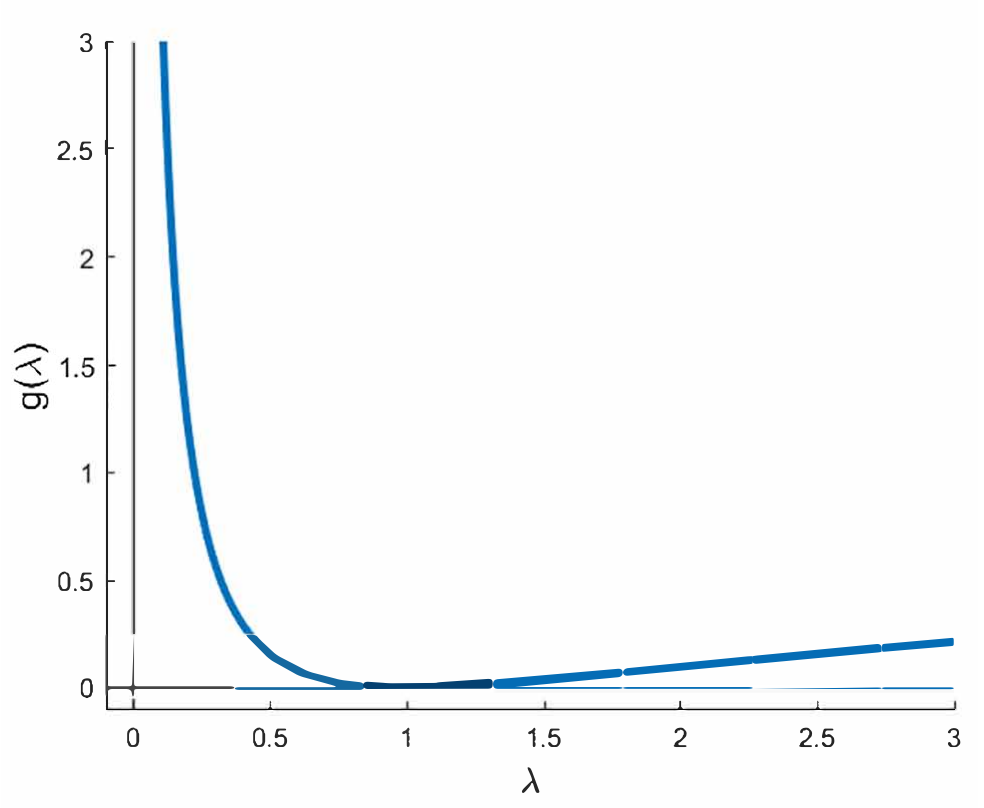}
    \caption{The $g(\cdot)$ function in (\ref{eq:g}).}
    \label{fig:g}
\end{figure}
Since the generalized eigenvalues of $(\Sigma_2,\Sigma_1)$ and $(\Sigma_1,\Sigma_2)$ are reciprocals of each other, $\lambda_i\geq 1$ implies $\lambda_i'\leq 1$. For case 1,  the generalized eigenvectors corresponding to the $r$ largest eigenvalues $\lambda_i$ are the same set of eigenvectors corresponding to the $r$ smallest eigenvalues $\lambda_i'$. With $\lambda_i'\leq 1$ for all $i$, these are precisely the basis of the optimal subspace when $D(q_2\|q_1)$ is to be maximized. Case 2 can be similarly proved.
\end{proof}
An important case for this class of problems is the classical `signal plus noise' model, as stated below.
\begin{corollary}
    Let $s\sim \mathcal{N}(0,\Sigma_s)$ and $n\sim \mathcal{N}(0,\Sigma_n)$. Furthermore, $s$ and $n$ are independent of each other. Let the sample $x$ be either $x=n$ (noise only) or $x=s+n$ (signal plus noise). Then the optimal $r$-dimensional subspace that maximizes $D(n\|s+n)$ or $D(s+n\|n)$ coincides with each other, where $n$ and $s+n$ in the KLD expressions denote the respective distributions of the reduced-dimension sample for the noise only and signal plus noise cases.  \label{cor:signalplusnoise}
\end{corollary} 
\begin{proof}
    This is a direct consequence of Theorem~\ref{lem:generalized}. The two distributions, noise only and signal plus noise are respectively $\mathcal{N}(0,\Sigma_n)$ and $\mathcal{N}(0,\Sigma_s+\Sigma_n)$. Since $\Sigma_s+\Sigma_n\succeq \Sigma_n$, where $A\succeq B$ means $A-B$ is positive semi-definite, the generalized eigenvalues of the matrix pair $(\Sigma_s+\Sigma_n,\Sigma_n)$ are all greater than or equal to $1$. 
\end{proof}

\section{Linear Projection for the General Case \label{sec:new}}

We now propose two new linear projection approaches for the general Gaussian model, i.e., the Gaussian distributions do not share a common mean or covariance matrix. The two approaches cater to different regimes in which the discriminative information is dominated by either the mean difference of the covariance matrix difference. The first one, developed for the so-called large-$\mu$ regime, is motivated by the fact that LDA, which is a one-dimensional linear projection, preserves the entire KLD for the equal covariance case with two classes. The second approach, developed for the so-called small-$\mu$ regime, relies on the additivity of the KLD with independent observations.
\subsection{The Large $\mu$ Regime}

A Gaussian distribution is completely specified by the first and second-order moments. 
A close examination of (\ref{eq:kld}) reveals that the contribution of the first and second-order moments can be loosely grouped into the following two terms: 
\begin{eqnarray} 
    D_\mu(p_1\|p_2)&=&\frac12 (\mu_2-\mu_1)^T\Sigma_2^{-1}(\mu_2-\mu_1),\label{eq:kldmu}\\ D_\Sigma(p_1\|p_2)&=&\frac12 \left[\ln \frac{|\Sigma_2|}{|\Sigma_1|} - d + \mbox{tr}(\Sigma_2^{-1}\Sigma_1)\right]. \label{eq:kldSigma} 
\end{eqnarray}
While the first term (\ref{eq:kldmu}) depends on the mean vectors $\mu_1, \mu_2$ and the covariance matrix $\Sigma_2$, the second term is only a function of the two covariance matrices. This motivates the two proposed approaches for the general case depending on whether or not $D_\mu(p_1\|p_2)$ dominates $D_\Sigma(p_1\|p_2)$. 

When the mean difference is large and the contribution from $D_\mu(p_1\|p_2)$ in (\ref{eq:kldmu}) is significant, a sensible approach is to find a subspace that completely preserves $D_\mu(p_1\|p_2)$ while retaining as much as possible the contributions due to the difference in covariance matrices, i.e., the term (\ref{eq:kldSigma}). Comparing (\ref{eq:kldequalSigma}) and (\ref{eq:kldmu}) and in light of Lemma~\ref{lem:LDA}, it is apparent that with two classes, a single dimension along $a=\Sigma_2^{-1}(\mu_2-\mu_1)$ is sufficient to retain the entire $D_\mu(p_1\|p_2)$ in (\ref{eq:kldmu}). This motivates the following algorithm for SDR under the so-called large-$\mu$ regime, i.e., when the contribution of $D_\mu$ is significant.

\begin{algorithm}
    Given $1\leq r< d$,

\begin{enumerate}
    \item Let $a_1=\Sigma_2^{-1}(\mu_2-\mu_1)$.
    \item For $a_2$ through $a_r$, choose the generalized eigenvectors for the matrix pair $(\Sigma_2,\Sigma_1)$ such that the corresponding generalized values maximize $g(\lambda)$ in (\ref{eq:g}). 
    \item The $r\times d$ linear projection matrix $A$ is defined as $A=[a_1,\cdots,a_r]^T$.
\end{enumerate}
\end{algorithm}
In the special case of equal covariance matrix, the contribution from (\ref{eq:kldSigma}) is $0$ and a single dimension along $a_1$ is sufficient. The direction $a_1$ coincides with the LDA with two classes (see Lemma~\ref{lem:LDA}).
Note that such a defined $A$ matrix is not necessarily orthonormal. QR factorization \cite{Golub:book} can be applied to find a set of basis vectors corresponding to the subspace spanned by $[a_1,\cdots,a_r]$. Projecting onto orthonormal basis of the subspace is especially desirable for visualization and exploratory data analysis.

\subsection{The Small $\mu$ Regime}

If the contribution from (\ref{eq:kldmu}) is not significant compared with that of (\ref{eq:kldSigma}), an alternative approach can be developed that takes advantage of the additivity of the KLD with independent observations. 

\begin{theorem} 
Let the eigenvalues and eigenvectors for the matrix $\tilde\Sigma=\Sigma_1^{-\frac12}\Sigma_2\Sigma_1^{-\frac12}$ be $(\lambda_i,u_i)$, $i=1,\cdots,d$. Define $\mu=\Sigma_1^{-\frac12}(\mu_2-\mu_1)$. Then the KLD between $p_1\sim \mathcal{N}(\mu_1,\Sigma_1)$ and $p_2\sim\mathcal{N}(\mu_2,\Sigma_2)$ can be decomposed into the sum of $d$ KLDs between univariate Gaussian distributions $\mathcal{N}(0,1)$ and $\mathcal{N}(u_i^T \mu,\lambda_i)$, i.e.,
\[
D(p_1\|p_2)=\sum_{i=1}^d D(\mathcal{N}(0,1)\|\mathcal{N}(u_i^T \mu,\lambda_i)).
\]
Each KLD term in the summation is attained by projecting the sample $\Sigma_1^{-\frac12} (x-\mu_1)$ into the one-dimensional subspace defined by the corresponding eigenvector $u_i$. \label{thm:additivity}
\end{theorem}

\begin{proof}
    The additivity property of the KLD states that for $n$ independently distributed random variables $x_1,\cdots,x_n$ whose distributions follow either
\[
p(x_1,\cdots,x_n)=\prod_{i=1}^n p_i(x_i)
\]
or 
\[
q(x_1,\cdots,x_n)=\prod_{i=1}^n q_i(x_i),
\]
the KLD between $p$ and $q$ equals the sum of individual KLDs, i.e., 
\[
D(p(x_1,\cdots,x_n)\|q(x_1,\cdots,x_n) = \sum_{i=1}^n D(p_i(x_i)\|q_i(x_i)).
\]
Consider a random sample $x\in R^d$ distributed according to either $\mathcal{N}(\mu_1,\Sigma_1)$ or $\mathcal{N}(\mu_2,\Sigma_2)$. Define $y=\Sigma_1^{-\frac12} (x-\mu_1)$, then $y$ follows either
\[
q_1\sim \Nmat(0,I) \quad \mbox{ or } \quad 
q_2 \sim \Nmat(\mu,\tilde\Sigma),
\]
where $\mu=\Sigma_1^{-\frac12}(\mu_2-\mu_1)$ and $\tilde\Sigma=\Sigma_1^{-\frac12}\Sigma_2\Sigma_1^{-\frac12}$. Let the eigen decomposition of $\tilde\Sigma$ be $\tilde\Sigma = U\Lambda U^T$ with $\Lambda=\mbox{diag}(\lambda_1,\cdots,\lambda_d)$ and $U=[u_1,\cdots,u_d]$. Define $z=U^Ty$, then $z$ follows either 
\[
r_1\sim \Nmat(0,I) \quad \mbox{ or } \quad 
r_2 \sim \Nmat(U^T\mu,\Lambda).
\]
Therefore, $z$ is an independent Gaussian vector under either $r_1$ or $r_2$.
The KLD between $r_1$ and $r_2$ can thus be decomposed into the sum of $d$ KLDs between univariate normal distributions $\mathcal{N}(0,1)$ and $\mathcal{N}(u_i^T \mu,\lambda_i)$, i.e.,
\begin{equation}
D(r_1\|r_2)=\sum_{i=1}^d D(\mathcal{N}(0,1)\|\mathcal{N}(u_i^T \mu,\lambda_i)). \label{eq:KLDsum}
\end{equation}
Since the transformations $y=\Sigma_1^{-\frac12} (x-\mu_1)$ and $z=U^Ty$ are both invertible given that $\Sigma_1$ and $U$ are both full rank, $D(p_1\|p_2)=D(q_1\|q_2)=D(r_1\|r_2)$.
\end{proof}

A natural way to pick a rank-$r$, $r<d$, subspace is to rank order the $d$ rank-$1$ subspaces defined by $u_i$'s according to $D(\mathcal{N}(0,1)\|\mathcal{N}(u_i^T \mu,\lambda_i)$ and choose the first $r$ eigenvectors, which corresponding to the largest $r$ terms in (\ref{eq:KLDsum}). 
If $\mu_1=\mu_2$, the solution reduces to the optimal solution described in Lemma~\ref{lem:generalized}. Additionally, for the general case with unequal means and covariance matrices, the KLD of the projected subspace is guaranteed to approach the full KLD of the original distributions as $r$ approaches $d$.

The detailed steps for the proposed algorithm are as follows.
\begin{algorithm}
  Given $1\leq r< d$,

\begin{enumerate}
    
    \item Compute $\mu=\Sigma_1^{-\frac12}(\mu_2-\mu_1)$.
    \item Compute $(\lambda_i,u_i)$, $i=1,\cdots,d$, the eigenvalue/eigenvector pairs for $\Sigma_1^{-\frac12}\Sigma_2\Sigma_1^{-\frac12}$.
    \item Compute, for $i=1,\cdots,d$, the $d$ component KLD:
    \[
     D(\mathcal{N}(0,1)\|\mathcal{N}(u_i^T \mu,\lambda_i))=\frac12\left(\ln \lambda_i - 1 +\frac{1+|u_i^T\mu|^2}{\lambda_i}\right).
    \]
    \item Rank order $(\lambda_i,u_i)$ according to $D(\mathcal{N}(0,1)\|\mathcal{N}(u_i^T \mu,\lambda_i))$.
    \item The projection matrix, applied to the transformed sample $y=\Sigma_1^{-\frac12}(x-\mu_1)$, is constructed by $r$ eigenvectors $u_i$'s corresponding to the largest $D(\mathcal{N}(0,1)\|\mathcal{N}(u_i^T \mu,\lambda_i))$.
\end{enumerate}
\end{algorithm}

The $r$-dimensional subspace found using Algorithm 2 is limited to that spanned by the eigenvectors of the matrix $\tilde{\Sigma}$. This is the consequence of decomposing the KLD in the form of (\ref{eq:KLDsum}). Such a decomposition does not take into account the actual mean separation even though the individual KLD depends on the mean difference vector. Therefore, the algorithm does not necessarily find the optimal subspace for any given $r$. 

With Algorithm 2, the whitening transformation with respect to $p_1$ puts the samples in an isotropic position if the distribution is indeed $p_1$. Projection onto $R^r$ using the constructed projection matrix maintains the isotropic position in $R^r$ since the eigenvectors of $\tilde\Sigma$ are orthonormal to each other. 

Note that for the special case of equal covariance matrix, i.e., $\Sigma_1=\Sigma_2=\Sigma$, the matrix $\Sigma_1^{-\frac12}\Sigma_2\Sigma_1^{-\frac12}$ is an identity matrix thus all its eigenvalues are $1$. In this case, any orthonormal basis can be the set of generalized eigenvectors. If we choose the first eigenvector to be precisely $\mu/|\mu|$, then the first component KLD in Step 3 simplifies to 
\[
\frac12\mu^T\mu=\frac12 (\mu_2-\mu_1)^T \Sigma^{-1} (\mu_2-\mu_1).
\]
This completely recovers the KLD for this special case. We note here that when implementing this algorithm, standard routines will typically return the computational basis as the eigenvectors instead of $\mu/|\mu|$ when the two covariance matrices are identical. 

\subsection{Determination of Large-$\mu$ or Small-$\mu$ Regimes}

Algorithms 1 and 2 are developed under two different regimes: Algorithm 1 is suitable for data when the mean separation is large relative to separation in any direction of the variances of the two classes. Algorithm 2 on the other hand is developed when the mean separation does not dominate the separation of covariance matrices along any direction. 

For a given $r<d$, i.e., the projected subspace is of dimension $r$, preserving the KLD term due to the mean difference (\ref{eq:kldmu}) requires only a single dimension. Thus a simple rule of thumb is to choose Algorithm 1 when 
\[
D_\mu \geq \frac{1}{r-1} D_\Sigma,
\]
and choose Algorithm 2 otherwise. We note that Algorithm 2 still captures information contributed by the mean difference in each step. When dealing with real data, cross-validation can again be used when selecting which algorithm is better for a given dataset. 

The difference between the two algorithms in terms of achieved KLD is more pronounced with small $r$. In fact, as $r$ becomes large, their performance, i.e., the recovered KLD, becomes close to each other. This is not surprising - supervised dimension reduction is useful when there is a concentration of discriminative information in a low-dimensional subspace. The fact that both algorithms choose projection directions in some descending order in recovered KLD suggests that when $t$ approaches $d$, both algorithms will recoup nearly the entire KLD of the original sample.

\section{$K$ classes \label{sec:K}}

For the case with $K>2$ classes of Gaussian populations, pairwise divergence measures can be used to search for low-dimensional representations for the classification problem. Both algorithms can be extended to the general $K>2$ case - ranking the subspaces can take different variations if classes are associated with different priorities. 

A special case that admits a concrete solution is when all the Gaussian distributions under the $K$ classes share the same covariance matrix, say, $\Sigma$\cite{Vogelstein:21}. For this special case, multiclass LDA, originally proposed in \cite{Rao:48}, is a popular means for supervised dimension reduction as it maximizes the overall mean separation along $K-1$ directions. Specifically, define
    \begin{align}
        \bar\mu&=\frac1K \sum_{k=1}^K\mu_k,\\
        S_\mu&=\sum_{k=1}^K (\mu_k-\bar\mu)(\mu_k-\bar\mu)^T. \label{eq:Smu}
    \end{align}
Multiclass LDA finds $K-1$ directions that maximize the between-class variance relative to the common variance: 
\begin{align}
    \max_{w} \frac{w^TS_\mu w}{w^T \Sigma w}. \label{eq:Kclass}
\end{align}
With $S_\mu$ being a rank-$(K-1)$ matrix, the solutions of $w$ are the $K-1$ generalized eigenvectors of the matrix pair $(S_\mu,\Sigma)$ corresponding to the $K-1$ non-zero generalized eigenvalues. Denote by those vectors $w_k$, $k=1,\cdots,K-1$, the $(K-1)$-dimensional subspace obtained by multiclass LDA is therefore
\begin{align}
\mathcal{S}_{LDA}={\rm Span}(w_1,\cdots,w_{K-1}). \label{eq:SLDA}
\end{align}

On the other hand, the pairwise KLD with common covariance matrix between classes $i$ and $j$ is simply 
\[
D(p_i\|p_j)=\frac12 (\mu_i-\mu_j)^T\Sigma^{-1}(\mu_i-\mu_j).
\]
We have the following simple result which states that the subspace $\mathcal{S}_{LDA}$ preserves all pairwise KLD.

\begin{theorem}
Multi-class LDA preserves the pairwise KLD of the original samples. \label{thm:K}
\end{theorem}

\begin{proof} The proof consists of two steps. In the first step, we construct a $(K-1)$-dimensional subspace and prove that it is KLD-preserving, i.e., the sample projected into this subspace preserves pairwise KLD of the original sample. The second step establishes that the constructed subspace is equivalent to the multiclass KLD in which the $(K-1)$-dimensional subspace is found by solving the maximization problem (\ref{eq:Kclass}).   

\noindent \underline{Step 1} 

Define, for $k=1,\cdots,K-1$, 
     \begin{align}
         a_k=\Sigma^{-1}(\mu_{k+1}-\mu_1).
    \end{align}
Furthermore, define $\mathcal{S}_{KLD}={\rm Span}(a_1,\cdots,a_{K-1})$, the space spanned by $(a_1,\cdots,a_{K-1})$. From Lemma~\ref{lem:LDA}, it suffices to show that for any $i\neq j$, $1\leq i,j\leq K$, $\Sigma^{-1}(\mu_i-\mu_j)\in \mathcal{S}_{KLD}$ since $\Sigma^{-1}(\mu_i-\mu_j)$ preserves $D(p_i\|p_j)$. 

The case when either $i=1$ or $j=1$ is trivial given Lemma~\ref{lem:LDA} and the definition of $a_k$. If $i\neq 1$ and $j\neq 1$, then both $a_i$ and $a_j$ lie in the subspace $\mathcal{S}_{KLD}$, so does $a_i-a_j$. But 
\[
\Sigma^{-1}(\mu_i-\mu_j)=\Sigma^{-1}(\mu_i-\mu_1)-\Sigma^{-1}(\mu_j-\mu_1)=a_i-a_j.
\]
Therefore, $\Sigma^{-1}(\mu_i-\mu_j)\in\mathcal{S}_{KLD}$.

\noindent \underline{Step 2} 

 We now establish that the above-defined subspace coincides with the subspace obtained through multi-class LDA, i.e.,
\begin{align}
    \mathcal{S}_{KLD}=\mathcal{S}_{LDA},
     \label{eq:LDAKLD}
\end{align}
where $\mathcal{S}_{LDA}$ is defined in (\ref{eq:SLDA}).

With $\Sigma$ being a full-rank matrix, let us define $y=\Sigma^{-\frac12}x$. Then the distributions for $K$ classes with sample $y$ are respectively $\mathcal{N}(\gamma_k,I)$, for $k=1,\cdots,K$, where $\gamma_k\triangleq \Sigma^{-\frac12}\mu_k$. 
Define 
\begin{align}
    S_\gamma = \sum_{k=1}^{K} (\gamma_k-\bar\gamma)(\gamma_k-\bar\gamma)^T,
\end{align}
where 
\[
\bar\gamma = \frac1K \sum_{k=1}^K \gamma_k.
\]

From the definition of $S_\mu$ in (\ref{eq:Smu}), we have $S_\gamma=\Sigma^{-\frac12}S_\mu \Sigma^{-\frac12}$.  Define $v_k\triangleq\Sigma^{\frac12}w_k$, for $k=1,\cdots, K-1$.
Since $(\lambda_k,w_k)$, $k=1,\cdots,K-1$, are the set of generalized eigenvalues/eigenvectors for the matrix pair $(S_\mu,\Sigma)$ corresponding to the nonzero $\lambda_k$'s,   $(\lambda_k,v_k)$'s are the set of eigenvalue/eigenvector pairs corresponding to the non-zero eigenvalues of $S_\gamma$.

Define $b_k=\gamma_{k+1}-\gamma_1$, $k=1,\cdots,K-1$, and with the definitions of $\gamma_k$ and $a_k$, we have $b_k=\Sigma^\frac12 a_k$, $k=1,\cdots,K-1$.  Proving (\ref{eq:LDAKLD}) is equivalent to proving
\begin{align}
    {\rm Span}(b_1,\cdots,b_{K-1})={\rm Span}(v_1,\cdots,v_{K-1}). \label{eq:colspace}
\end{align}

From the definition of $S_\gamma$, the set of vectors $(\gamma_k-\bar\gamma)$, $k=1,\cdots,K-1$, span the column space of $S_{\gamma}$. In fact, given $\sum_{k=1}^K (\gamma_k-\bar{\gamma}) = 0$, the column space $S_\gamma$ is spanned by $(\gamma_k-\bar{\gamma})$, $k=2,\cdots, K$. But for $k=2,\cdots, K$, 
\begin{align*}
    \gamma_k-\bar{\gamma}&= \gamma_k-\gamma_1 - \frac1K \sum_{l=1}^K (\gamma_l-\gamma_1) \\
    &= b_{k-1}-\frac1K \sum_{l=2}^K b_{l-1}.
\end{align*}
Therefore the column space of $S_\gamma$ is also spanned by the set of vectors $(b_1,\cdots,b_{K-1})$.
Then equation (\ref{eq:colspace}) is established since  the eigenvectors $(v_1,\cdots,v_{K-1})$ also span the column space of the matrix $S_\gamma$
\end{proof}
If the objective is to preserve pairwise KLD, then with the subspace obtained in Theorem~\ref{thm:K}, there is no need to add additional dimensions. That is, for $K$ classes of Gaussian populations with equal covariance matrix, a $(K-1)$-dimensional subspace is sufficient to completely retain pairwise KLD of the original samples. The same statement applies to pairwise CI as the two metrics are equivalent when all classes share a common covariance matrix.

\section{Numerical Experiments \label{sec:numerical}}
Two sets of numerical experiments are conducted. The first one uses randomly generated {\em Gaussian parameters} to examine the behavior of the proposed linear projection methods for SDR. 
Specifically, we embed a low-dimensional Gaussian signal in a high-dimensional sample space and then project the sample back to a lower-dimensional space. Visualizing the probability density functions of the reduced-dimensional samples and computing the corresponding KLDs for different projection methods provide the reader with a concrete and intuitive way to compare the performance of different SDR approaches. In the second part, actual Gaussian samples are generated for the two classes in the high-dimensional space, and their reduced-dimension representations are used to evaluate the recovered KLD and classification performance. The projection matrix for the second part is computed using the sample statistics obtained with the samples generated in the high-dimensional space. Similarly, the KLD computation in the reduced-dimension subspace uses the sample statistics, i.e., the estimated mean vectors and covariance matrices of the projected samples. 

In addition to Algorithms 1 and 2, we implement the gradient descent algorithm that directly searches for the projection matrix $A$ to maximize (\ref{eq:reducedkld}). The gradient descent algorithm is implemented using Tensorflow's AdamOptimizer class and the learning rate is empirically chosen that exhibits the best convergence property for the given set of parameters. Different from the gradient descent approach used in \cite{Dwivedi:22}, we do not impose an orthogonality constraint on the projection matrix after each iteration. From Lemma~\ref{lem:matrix}, any matrix spanning the same row space achieves the same KLD, and if orthogonality is desired for the final projection matrix, QR factorization can be applied to the converging matrix $A$. Our extensive experiments suggest that the most critical issue affecting the performance of the gradient descent algorithm is its initialization. As it turns out, for all the cases that we have tested, random initialization never outperforms the case when we use the solution to either Algorithm 1 or Algorithm 2 as the initial matrix, indicating that either or both of the proposed algorithms are potentially close to the global optimum, as least in the cases that we have tested.

In the second part, in addition to Algorithms 1 and 2, we also include LoL \cite{Vogelstein:21} in our evaluation. The LoL was proposed by assuming an equal covariance matrix between the two classes. When applying to the case with unequal covariance matrices, the entire samples, after re-centering using the respective sample means, are used to estimate a `common' covariance matrix. While we have established that if the samples for the two classes share a common covariance matrix, projection onto a one-dimensional subspace along the LDA direction suffices to recover the entire KLD, i.e., adding additional dimensions does not provide any benefit for discriminating the two classes, with unequal covariance matrices, adding additional dimensions will improve the performance since in this case a single dimension is inadequate to capture the entire KLD between the two classes.    

For both parts, the first using only distribution parameters and the second using randomly generated samples, the final projected subspace is in $R^2$. This enables the visualization of either the density function or the sample scatter plot.

\subsection{Density separation}
Consider a situation when a random signal $s\in R^t$ takes one of two possible Gaussian distributions. Specifically,  the $t\times 1$ signal $s$ is  distributed according to either
\[
p_{s1}\sim \mathcal{N}(\mu_1,\Sigma_1), \quad \mbox{ or } \quad p_{s2}\sim \mathcal{N}(\mu_2,\Sigma_2).
\]
The signal is not directly observable but is embedded in $R^d$ through a linear Gaussian channel, resulting in the sample $x\in R^d$. 
\[
x=Hs+z
\]
where $H$ is a $d\times t$ full column rank channel matrix and $z\sim \mathcal{N}(0,\sigma^2 I_d)$. The signal $s$ and noise $z$ are independent of each other. Thus $x$ is distributed according to either 
\begin{align}
    p_{x1}\sim \mathcal{N}(H\mu_1,H\Sigma_1H^T+\sigma^2 I_d), \quad \mbox{ or } \quad p_{x2}\sim \mathcal{N}(H\mu_2,H\Sigma_2H^T+\sigma^2 I_d). \label{eq:x}
\end{align}

Dimension reduction on $x$ would then bring $x$ back to $R^r$ for some $r<d$, giving rise to an $r$-dimensional sample $y$. Let the $r\times d$ linear projection matrix be $A$ and thus $y=Ax$. Then $y$ is distributed according to either 
\begin{align*}
    p_{y1}&\sim \mathcal{N}(AH\mu_1,A(H\Sigma_1H^T+\sigma^2 I_d)A^T)\mbox{ or} \\ p_{y2}&\sim \mathcal{N}(AH\mu_2,A(H\Sigma_2H^T+\sigma^2 I_d)A^T).
\end{align*} 

With the distribution parameters of $x$ in (\ref{eq:x}), Algorithms 1 and 2, along with the gradient descent algorithm, are used to find the linear projection from $R^d$ to $R^r$. We evaluate the dimension reduction performance through visualization of the projected density functions in $R^r$ and the attained KLD by the different projection algorithms. For visualization, we will restrict $r=2$ whereas for KLD evaluation, we compute the recovered KLD as a function of varying $r$. 

We now define the three KLDs below, each calculated using (\ref{eq:kld}) with the corresponding statistics. 
\begin{align}
    D_s&=D(p_{s1}\|p_{s2}),  \label{eq:KLDs}\\
    D_x&=D(p_{x1}\|p_{x2}). \label{eq:KLDx}\\
    D_y&=D(p_{y1}\|p_{y2})  \label{eq:KLDy}
\end{align}
From the DPI for KLD, $D_y\leq D_x \leq D_s$.

In our experiment, we choose $t=10$ and $d=100$. For visualization, we set $r=2$ while for KLD evaluation, $r$ varies from $1$ to $10$. 
The Gaussian parameters for $s$ are randomly generated while enforcing the positive definiteness of the covariance matrices. Two different sets of parameters are selected, one corresponding to the large-$\mu$ regime while the other to the small-$\mu$ regime. The determination of the regime is based on the values of $D_\mu$ and $D_\Sigma$ in (\ref{eq:kldmu}) and (\ref{eq:kldSigma}) computed using the statistics for the high dimensional sample $x$ under the two distributions.  For this example, we set the noise variance to be $1$ and the $(D_\mu,D_\Sigma)$ pair for $x\in R^d$ takes the values $(778.4,168.2)$ for the large-$\mu$ case and $(1.7,220.8)$ for the small-$\mu$ case. 

\begin{figure}[h!]
    \centering
    \begin{tabular}{cccc}
    \includegraphics[width=1.5in, height=1.3in]{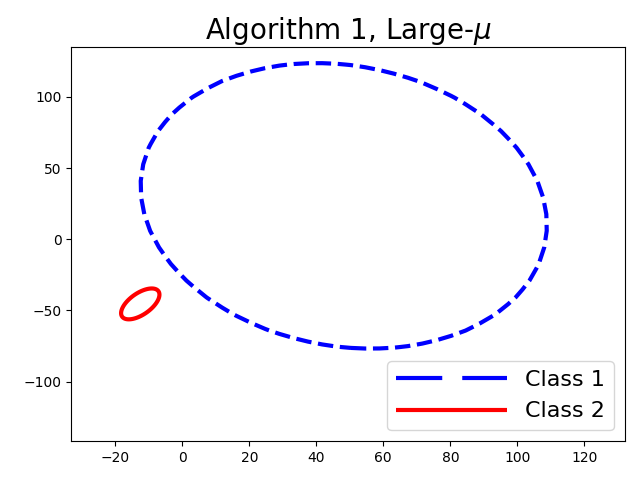}&
    \includegraphics[width=1.5in, height=1.3in]{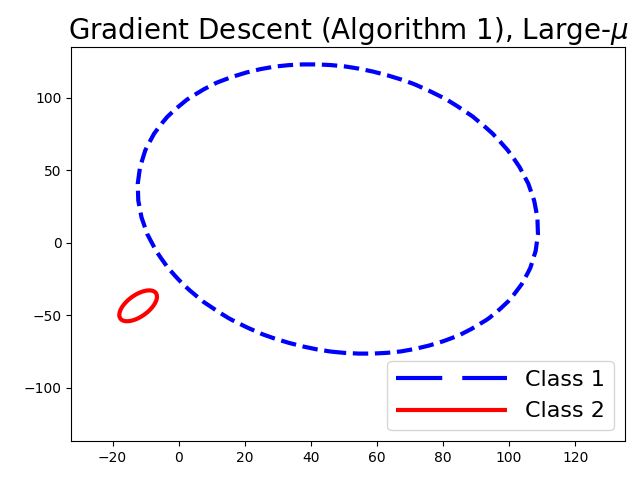}&
    \includegraphics[width=1.5in, height=1.3in]{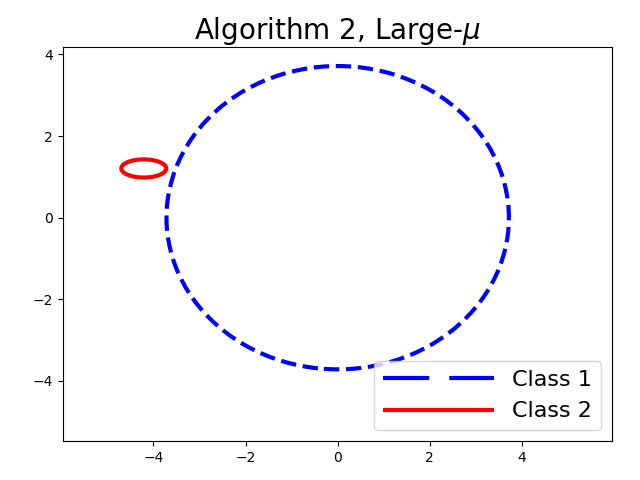}&
        \includegraphics[width=1.5in, height=1.3in]{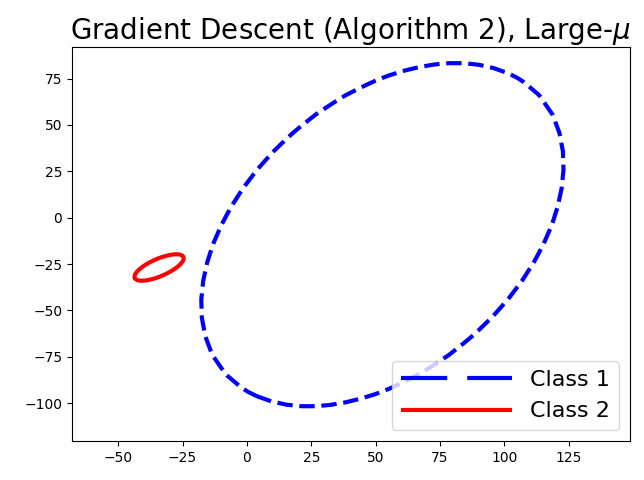}\\
   (a) & (b) &(c) & (d)
    \end{tabular}
    \caption{Density plots of the projected samples using the two proposed algorithms for the large-$\mu$ case: (a) Algorithm 1; (b) Gradient descent using Algorithm 1 as initialization; (c) Algorithm 2; and (d) Gradient descent using Algorithm 2 as initialization. The contour location is chosen at $1/1000$ of the peak density value. }
    \label{fig:densitylargemu}
\end{figure}

\begin{figure}[h!]
    \centering
    \begin{tabular}{cccc}
    \includegraphics[width=1.5in, height=1.3in]{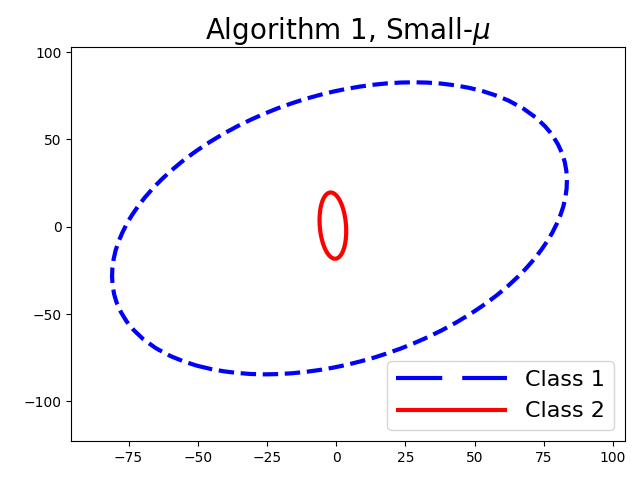}&
    \includegraphics[width=1.5in, height=1.3in]{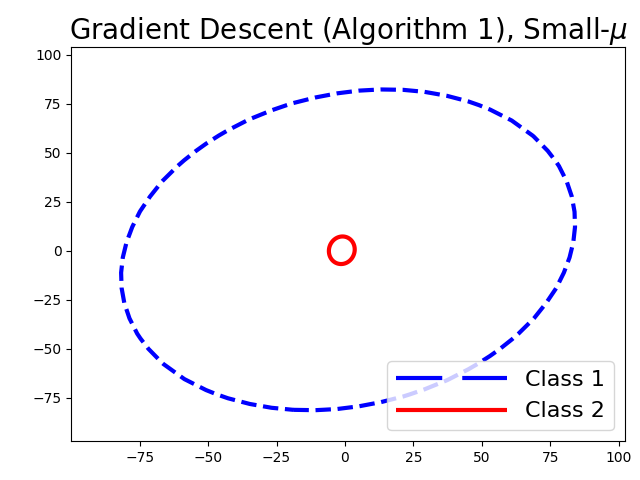}&
    \includegraphics[width=1.5in, height=1.3in]{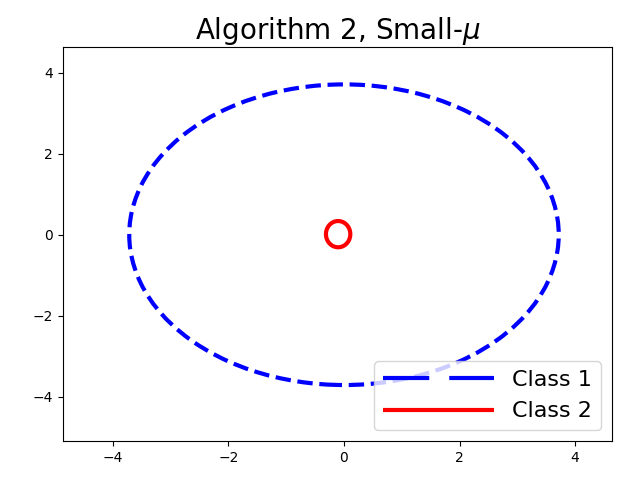}&
        \includegraphics[width=1.5in, height=1.3in]{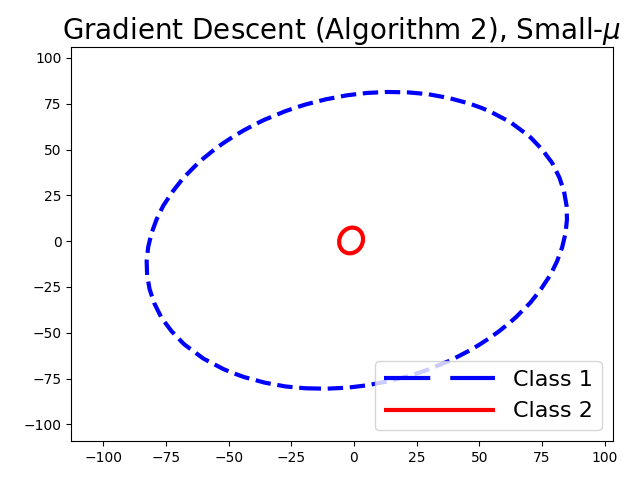}\\
   (a) & (b) &(c) & (d)
    \end{tabular}
    \caption{Density plots of the projected samples using the two proposed algorithms for the small-$\mu$ case: (a) Algorithm 1; (b) Gradient descent using Algorithm 1 as initialization; (c) Algorithm 2; and (d) Gradient descent using Algorithm 2 as initialization. The contour location is chosen at $1/1000$ of the peak density value. }
    \label{fig:densitysmallmu}
\end{figure}

The projected density functions in $R^2$ are plotted in Fig.~\ref{fig:densitylargemu} for the large-$\mu$ case and Fig.~\ref{fig:densitysmallmu} for the small-$\mu$ case. In both cases, along with Algorithms 1 and 2, we plot the results using gradient descent initialized using with matrices found through Algorithms 1 and 2. Some interesting observations are as follows.  
\begin{itemize}
    \item With large-$\mu$, while all algorithms have a sufficient separation between the two classes, Algorithm 1 has a larger separation in the mean vectors because it utilizes its first dimension to capture the entire KLD contributed by the mean difference (notice that the scales of the two axes are different for Figs.~\ref{fig:densitylargemu}(a) and (c)). 
    \item For the small-$\mu$ case, the mean vectors largely coincide with each other - $D_\mu$ is $1.7$ and is much smaller compared with $D_\Sigma=220.8$ in this case. Algorithm 2 is notably better in that the contrast between the two density functions is much more pronounced. More specifically, the projected subspace in $R^2$ using Algorithm 2 results in a more pronounced difference in variances in both axes for the two classes. This is reflected in Fig.~\ref{fig:densitysmallmu}(c) where the class 2 density plot is much smaller relative to class 1 as compared with Fig.~\ref{fig:densitysmallmu}(a). 
    \item For Algorithm 2, the density function is a standard normal for class 1, i.e., its contour is a perfect circle. This is because Algorithm 2 puts class 1 samples in an isotropic position through a whitening transformation.  
    \item With initialization using Algorithm 1, gradient descent leads to density plots that largely track that of Algorithm 1 itself though there seems to be notable improvement for the small-$\mu$ case. 
    \item With initialization using Algorithm 2, gradient descent leads to density plots that seem to be drastically different from the initialization. However, the difference is largely in scaling. In addition, while the initialization using Algorithm 2 has class 1 put in an isotropic position, gradient descent does not enforce such a constraint. A detailed comparison of achieved KLD between the initialization and the convergent point with gradient descent is given below.
    \item For both large-$\mu$ and small-$\mu$ tested here, as well as other examples not reported here, random initialization never outperforms the initialization using either Algorithm 1 or Algorithm 2. While the experiments are not exhaustive, it appears that one or both of the two algorithms results in the projection matrix that is close to the global optimum for the cases that we have tested.

\end{itemize}


Tables~\ref{tab:kldlargemu} and \ref{tab:kldsmallmu} include the recovered KLD as $r$ varies from $1$ to $10$ for the four algorithms: Algorithms 1 and 2 and the corresponding gradient descent algorithms with either of the two used for initialization. 
It is apparent from the tables that while gradient descent does improve the KLD recovery for small $r$, its benefit vanishes as $r$ increases, indicating that Algorithm 1 and/or Algorithm 2 appear to give at least a local optimum. 



\begin{table}[h!]
 {\small    \centering
        \caption{Retained KLD as a function of the dimension of the projected space $r$ for the large-$\mu$ case. }
    \label{tab:kldlargemu}
    \begin{tabular}{|c|c|c|c|c|c|c|c|c|c|c|}
  \hline
        $r$ & $1$ & $2$ & $3$ &  $4$ & $5$ & $6$ & $7$ & $8$ & $9$ & $10$   \\ \hline 
Algorithm 1 &831.740&939.192&942.306&944.295 &945.282&945.795&946.063&946.335&946.557&946.568\\ \hline 
Algorithm 2 &542.069&883.577&915.543&930.837 &939.503&943.017&945.034&946.135&946.470&946.568 \\ \hline 
Gradient descent 1 &834.413 &939.462&942.387&944.390&945.381&945.795&946.063&946.335&946.557&946.568 \\  \hline 
Gradient descent 2 &677.258&939.150&941.771&942.530 &943.013&945.027&945.340&946.329&946.557&946.568 \\  \hline 
    \end{tabular}}
\end{table}

\begin{table}[h!]
{\small    \centering
        \caption{Retained KLD as a function of the dimension of the projected space $r$ for the small-$\mu$ case. }
    \label{tab:kldsmallmu}
    \begin{tabular}{|c|c|c|c|c|c|c|c|c|c|c|}
  \hline
        $r$ & $1$ & $2$ & $3$ &  $4$ & $5$ & $6$ & $7$ & $8$ & $9$ & $10$   \\ \hline 
Algorithm 1 &144.118&160.425&217.176&219.524 &220.502&221.500&222.097&222.441&222.505&222.506\\ \hline 
Algorithm 2 &153.845&216.700&219.040&220.414 &221.429&222.035&222.373&222.491&222.506&222.506\\ \hline 
Gradient descent 1 &153.635 &216.512&218.075&220.416&221.430&222.036&222.154&222.492&222.506&222.506 \\  \hline 
Gradient descent 2 &153.846&216.701&219.041&220.416 &221.431&222.036&222.374&222.492&222.506&222.506 \\  \hline 
    \end{tabular}}
\end{table}


\subsection{Visualization and classification performance}

We now use synthetic data to evaluate the performance of 
proposed linear projection approaches, through visualization of projected samples using scatter plots and the classification performance. 

In this part, original high dimensional samples for both classes were generated with random means and covariance matrices in $R^6$, i.e., $d=6$. The generated samples were used to estimate the sample means and sample covariance matrices. These sample statistics are used in the proposed algorithms and LoL to find the projection matrix from $R^6$ to $R^2$. For this part, while we have also implemented gradient descent algorithms with different initialization, the results are not included since the results of the gradient descent algorithm with either Algorithm 1 or Algorithm 2 for initialization are nearly indistinguishable from that of the initialization point, indicating that both Algorithm 1 and Algorithm 2 are close to the local optimum. Additionally, gradient descent with either LoL or random initialization performs far worse than that using either Algorithm 1 or 2. 

The scatter plots for the projected samples in $R^2$ are plotted in Fig.~\ref{fig:fig:largemuvisual} for the large-$\mu$ case and Fig.~\ref{fig:smallmuvisual} for the small-$\mu$ case. Again, for Algorithm 2, class 1 samples were put through a whitening transformation hence the projected samples for class 1 have isotropic positions centered at the origin.  It is quite clear from the scatter plots that both Algorithms 1 and 2 significantly outperforms LoL. 
The difference is quite pronounced in that the subspace projected by Algorithms 1 and 2 lead to highly contrasting concentrations of the samples belonging to the two classes whereas the LoL lead to two classes that have significant overlap in the projected subspace. The reason that LoL does not perform well is two-folded. While it does include the mean difference vector in its projected subspace, it does not account for the difference in variances between the two classes along the same direction. The second reason is that the additional dimension added to the LoL subspace is the principal component of the common covariance matrix - merely having the largest variance along that direction by no means indicates its usefulness in discriminating between the two classes.
\begin{figure}
    \centering
    $\begin{array}{ccc}
    \includegraphics[width=2.4in]{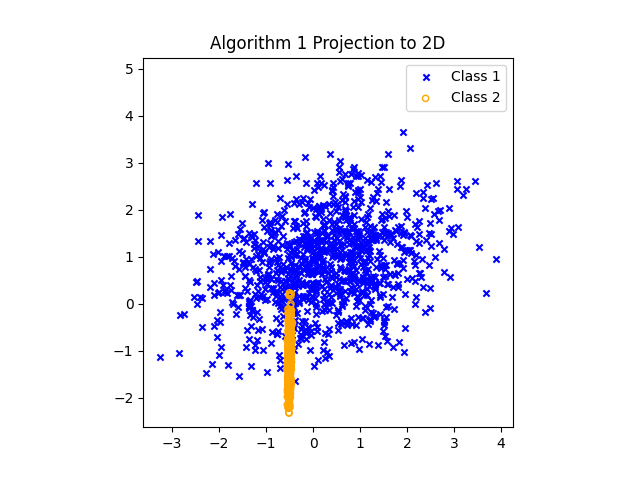}  &\hspace*{-0.5in}
    \includegraphics[width=2.4in]{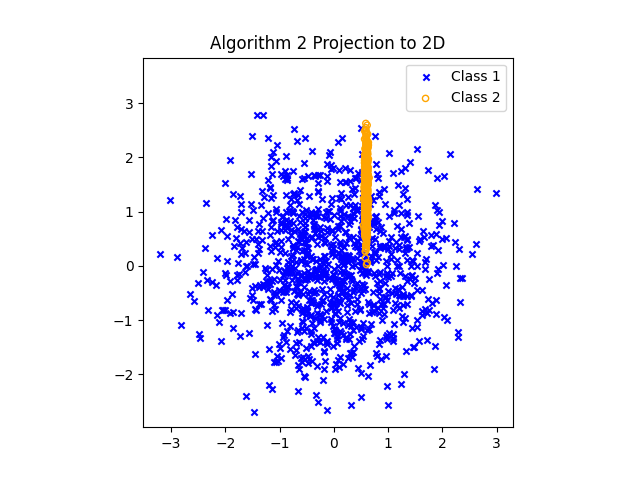} &\hspace*{-0.5in}
    \includegraphics[width=2.4in]{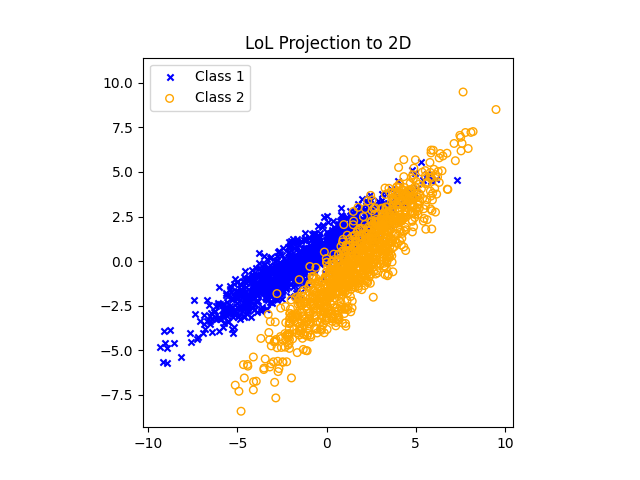}\\
   (a) & (b) &(c) 
    \end{array}$
    \caption{Scatter plots of the projected samples using the three SDR approaches for large-$\mu$: (a) Algorithm 1,  (b) Algorithm 2, (c) LoL.}
    \label{fig:fig:largemuvisual}
\end{figure}

\begin{figure}
    \centering
    $\begin{array}{ccc}
    \includegraphics[width=2.4in]{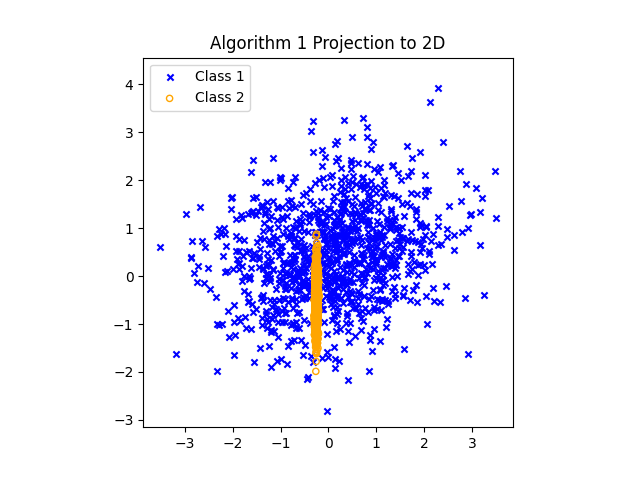} &\hspace*{-0.5in}
    \includegraphics[width=2.4in]{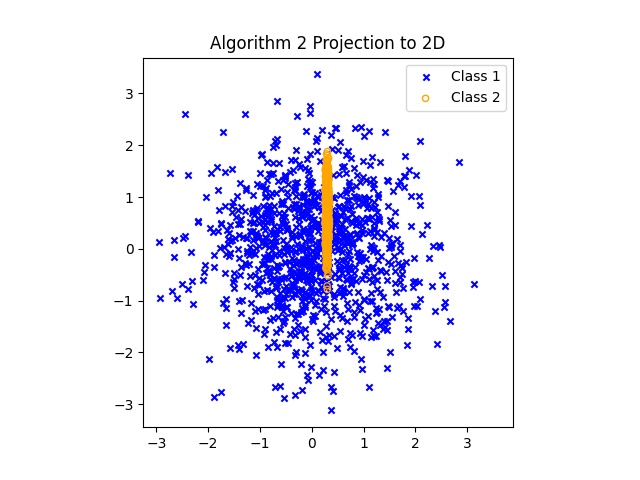} &\hspace*{-0.5in}
    \includegraphics[width=2.4in]{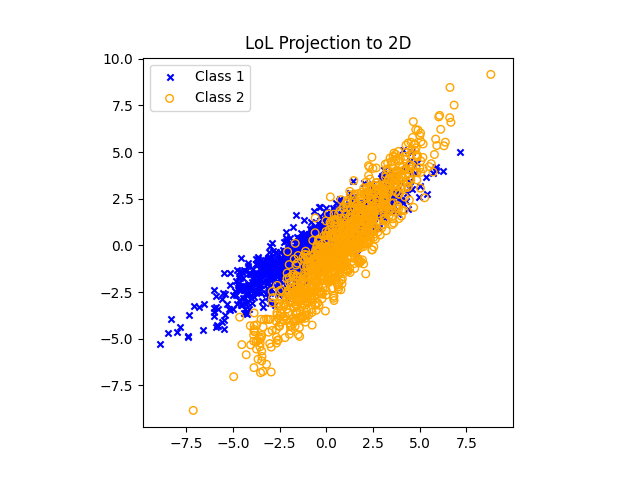}\\
    (a)&(b)&(c) 
    \end{array}$
    \caption{Scatter plots of the projected samples using the three SDR approaches for small-$\mu$: (a) Algorithm 1, (b) Algorithm 2, and (c) LoL.}
    \label{fig:smallmuvisual}
\end{figure}

Besides visualizing the scatter plots of the reduced-dimension samples in $R^2$, we also compute the recovered KLD as well as the classification performance using the reduced-dimension samples. The computation of all the KLD quantities uses the sample statistics instead of the true statistics that are used to generate the data. For classification, support vector classifier (SVC) from sklearn with radial basis function kernel is used. A total of $10,000$ samples per class are generated for training with $10$-fold cross-validation. With the trained model,  independently generated $1,000$ samples per class are used to evaluate the classification performance of the trained SVC. The results are summarized in Table~\ref{tab:kldsvc}.

\begin{table}[h!]
    \centering
        \caption{Comparison in KLD and classification accuracy of Algorithm 1, Algorithm 2, and LoL. }
    \label{tab:kldsvc}
    \begin{tabular}{|c|c|c|c|c|}
  \hline
    & \multicolumn{2}{c|}{Large-$\mu$} & \multicolumn{2}{c|}{Small-$\mu$} \\ \hline
         & Recovered KLD & Accuracy & Recovered KLD &  Accuracy \\ \hline 
Algorithm 1 &  3242.9& 96.00\%       & 2254.9 & 95.85\% \\ \hline 
Algorithm 2 &  3238.5 &  97.15\%     & 2253.9& 96.10\% \\ \hline 
LoL & 6.6 &   93.40\%    & 2.4 & 78.70\%  \\  \hline 
    \end{tabular}

\end{table}
For this example, the two proposed algorithms have comparable performance in achievable KLD and classification accuracy and both outperform LoL.
It is interesting to note that while KLD can be used as a proxy for classification performance, the magnitude of the KLD does not proportionally translate into the classification performance with finite numbers of samples.

\section{Conclusion and Future Work \label{sec:conclusion}}

This paper proposes two linear projection approaches for supervised dimension reduction using only first and second-order statistics. The proposed algorithms are derived by maximizing the Kullback-Leibler Divergence (KLD) for the general Gaussian model. They subsume existing linear projection solutions developed under some simplifying assumptions of Gaussian distributions. As a by-product, we establish that multiclass Linear Discriminant Analysis (LDA) is optimal in maximizing pairwise KLD with a common covariance matrix. The proposed algorithms are shown, both in theory and through numerical experiments, to outperform existing linear projection approaches under the most general Gaussian model.

Understanding how the proposed algorithms generalize to multiple classes under the general Gaussian model will significantly broaden their applications. We note that while the multiclass LDA \cite{Rao:48} assumes a common covariance matrix, the original Fisher's discriminant analysis (FDA)\cite{Fisher:36} was developed under the assumption that different classes have distinct mean and covariance matrices. FDA uses the heuristic criterion of maximizing the ratio of between-class variance and within-class variance along the projected direction. Thus, it provides a benchmark for any generalization of the proposed algorithms to the multiclass case.

A notable property of the proposed algorithms is that they require only the first and second-order statistics, are easy to implement, and admit ready interpretation. An important direction is to examine how the algorithms perform when dealing with real data that are typically non-Gaussian. The simplicity of the algorithms suggests that they should have reasonable generalization properties; this, however, requires extensive numerical experiments using real datasets, which will be an interesting direction for future work.

Another interesting direction is to devise supervised dimension reduction in a similar manner but for non-Gaussian populations. A particular case of interest is the Gaussian mixture with either common or distinct Gaussian components among different classes. With Gaussian mixture, the first and second order statistics are easy to compute, thus one can apply the proposed algorithms for dimension reduction. The performance provides a baseline for any proposed algorithms designed specifically for the Gaussian mixture model. While there is no closed-form expression for the KLD between Gaussian mixtures, various approximations have been developed in the literature \cite{Hershey:07} and can be used in place of the exact KLD for deriving dimension reduction solutions. 


\bibliography{ref}

\end{document}